\documentclass{article}
\usepackage{pifont}
\usepackage{amssymb}
\usepackage{stmaryrd}
\usepackage[pdftex]{graphicx}
 \usepackage{amsmath}

\author{}
\newtheorem{theorem}{Theorem}
\newtheorem{lemma}[theorem]{Lemma}
\newtheorem{definition}{Definition}
\newenvironment{proof}{{\sc Proof. }}{\vspace{0.2in}}

\begin{document}

\begin{center}
\vskip 1em
\begin{Large}
On the Information of \\[.2ex]the Second Moments Between Random Variables\\[1.2ex]Using Mutually Unbiased Bases
\end{Large}
\vskip 1em
{Hongyi Yao \footnotemark}\\[.6em]
\end{center}
\vskip 1em

 \footnotetext {Hongyi Yao is a PHD student of Institute of
Theoretical Computer Science of Tsinghua University, Beijing,
100084, P. R. China(email: thy03@mails.tsinghua.edu.cn)

Supported by National Natural Science Foundation of China under
grant No. $60553001$ and National Basic Research Program of China
under grant No. $2007CB807900$, $2007CB807901$.}

\begin{abstract}
The notation of mutually unbiased bases(MUB) was first introduced by
Ivanovic to reconstruct density matrixes\cite{Ivanovic}. The subject
about how to use MUB to analyze, process, and utilize the
information of the second moments between random variables is
studied in this paper. In the first part, the mathematical
foundation will be built. It will be shown that the spectra of MUB
have complete information for the correlation matrixes of finite
discrete signals, and the nice properties of them. Roughly speaking,
it will be shown that each spectrum from MUB plays an equal role for
finite discrete signals, and the effect between any two spectra can
be treated as a global constant shift. These properties will be used
to find some important and natural characterizations of random
vectors and random discrete operators/filters. For a technical
reason, it will be shown that any MUB spectra can be found as fast
as Fourier spectrum when the length of the signal is a prime number.

In the second part, some applications will be presented. First of
all, a protocol about how to increase the number of users in a basic
digital communication model will be studied, which has bring some
deep insights about how to encode the information into the second
moments between random variables. Secondly, the application of
signal analysis will be studied. It is suggested that complete "MUB"
spectra analysis works well in any case, and people can just choose
the spectra they are interested in to do analysis. For instance,
single Fourier spectra analysis can be also applied in nonstationary
case. Finally, the application of MUB in dimensionality reduction
will be considered, when the prior knowledge of the data isn't
reliable.

\end{abstract}


\begin{center}
\vskip 1em
 I. INTRODUCTION
 \vskip 1em
\end{center}
\vskip 1em
Ivanovic first introduced mutually unbiased bases(MUB) to
reconstruct density matrixes \cite{Ivanovic}:
\begin{definition}
Let $M_v=\{v_1,v_2,...v_d\}$, $M_u=\{v_1,v_2,...v_d\}$ be two
normalized orthogonal bases in the $d$ dimension complex space. They
are said to be mutually unbiased bases if and only if
$|<v_i,u_j>|=\frac{1}{\sqrt d}$, for any $i,j=1,2,...,d$. A set of
normalized orthogonal bases $\{M_1,M_2,...,M_n\}$ are said to be
mutually unbiased bases if and only if each pair of bases $M_i$ and
$M_j$ are mutually unbiased bases.
\end{definition}
MUB is widely used in the areas of quantum physics and quantum
information theory, such as the reconstruction of
pre-state\cite{wotters}, tomography, Wigner
distribution\cite{durt1}, teleportation\cite{durt2}, and quantum
cryptograph \cite{qc1,qc2,qc3}. But it has only a few classical
application such as \cite{CDMA}. This is quite reasonable, because
do full MUB spectra analysis need $d+1$ times time and space
resources where $d$ equals the length of signals. But it should be
noticed that bases from MUB has natural connections with the Fourier
base which has plenty of applications, \cite{Michel} has done some
study about it. Intuitively, the relation between any two bases from
MUB is the same as that between the standard bases and the Fourier
bases if we only concern the inner products of the vectors.

One of the major subjects in this area is to construct MUB for a
given dimension $d$. It's known that, there are no more than $d+1$
MUB for dimension $d$, and when $d$ is the power of prime, all $d+1$
MUB can be explicitly constructed\cite{wotters}. This paper only
focuses on the case when $d+1$ MUB can be found for dimension $d$,
and will not study the construction of the. It will be introduced,
in Sections $II-IV$, some mathematical foundations. Then the paper
will present some interesting applications of these results in
Sections $V-VII$.

In Section $II$, the equivalence between autocorrelation matrix and
the spectra of mutually unbiased based will be formally presented.
Some interesting properties concerning what kinds of spectra can
form autocorrelation matrix are studied, such as the generalization
of  Uncertainty Principle. It will be shown that the equivalent
relation is robust, because the effects of small errors are also
trivial.

In Section $III$, some nice properties of the spectra of MUB will be
studied. First, the original definition of "stationary" will be
extremely extended, and it's interesting to see that any discrete
random signals can have all kinds of "stationary" versions of them.
Then, the relationship between related random sources and
independent random sources will be presented, it will be shown that
treating normal random sources as a bunch of independent random
sources will bring a lot of convenience. Of course, MUB is the key
tool. The third part of this section is going to use the nice
properties of MUB to do complete analysis for random
operators/fielters. This part will introduce a general way to do all
kinds of stabilization for random vectors with some compensations on
"white noise". At last, a filter which only deal with some
designated spectra and left others untouched will be presented.

In Section $IV$, the MUB spectra for a deterministic vector will be
studied. In the first part, an algorithm will be shown which tells
that any MUB transform can be done as fast as DFT when in prime
dimensions. Then some properties of the MUB spectra for
deterministic vectors will be listed.

The main application of above results is an simple digital
communication protocol which can significantly increase the number
of users without any advanced techniques such as\cite{Trans}. This
will be introduced in Section $V$. Maybe the theoretical protocol is
far from practice, but it provides some deep insights about how to
encode information into the second moments between random variables
based on above results. Roughly speaking, communication using the
first moments of the signals is well studied\cite{DigitalCom}, while
our protocol is based on the moment of higher order. When some users
are idle, the protocol retrogresses simple ones such as
"TDMA"/"FDMA". Based on results of Section $III$, we will introduce
some interesting alternations of this model which suggest we can do
many things based on such model.

In Section $VI$, we study the application of signal analysis.
Spectra analysis for stationary signal is useful and well
known\cite{spectra,sta}, while nonstationary case are much
harder\cite{unsta1,unsta2}. Using MUB, we suggest complete spectra
analysis for discrete signals works well in any case. Actually, it
suggests that people can choose the spectra they are interested in
to do the analysis. For instance, Fourier spectra analysis also make
sense in nonstationary case. We will give a example about how to
apply it to signal detection. However, we should do more about the
physical meanings of the nonfourier spectra of MUB, because they are
important for practical and mathematical reasons.

Finally, we will consider the applications of MUB in dimensionality
reduction. In the case when no prior knowledge of the data is known,
we will present some local results and a global conjecture. When
the prior knowledge is not reliable, we suggest that MUB work well.

We will give some basic notations for the paper. We only work in $d$
dimension complex linear space, where the whole $d+1$ mutually
unbiased bases(MUB) can be found. Assume ${M_1,M_2,...M_{d+1}}$ are
the MUB of $d$ dimension complex linear space where the columns of
$M_i$ form the $i'th$ base. Without loss of generality, $M_1$ is the
standard base for dimension $d$ complex linear space. For all random
variables mentioned in this paper, the estimation values of them are
zero because constant shift is easy to handle. So in the paper,
autocorrelation matrixes has the same
 meaning of correlation matrixes. Each vector is
a vertical vector as default. $Rx$ is assumed the autocorrelation
matrix of complex random vector $X=\{x_1,x_2...x_d\}^T$ , and
$tr(Rx)=1$ as default. We say $x$ is "white noise" if and only if
$E(x)=0$ and $x$ is independent to all other random variables
mentioned in this paper.
\begin{center}
\vskip 1em
 II. THE EQUIVALENCE BETWEEN CORRELATION MATRIXES AND THE SPECTRA OF MUB
 \vskip 1em
\end{center}
\vskip 1em Ivanovic first introduced the idea about using the
spectra of mutually unbiased bases to reconstruct density matrixes
of quantum states \cite{Ivanovic}. It's easy to see that when apply
a unitary matrix $U$ to random vectors, the change of correlation
matrixes is the same as that for density matrixes when apply $U$ to
the quantum states. So follow the notations of introduction, we give
some basic definitions.

\begin{definition}
 Let $k$-Spectrum $S_k$ of $Rx$ be the diagonal part of matrix
$M_i^{H}\cdot Rx \cdot M_i$. And the set $\{S_1,S_2...,S_{d+1}\}$
form the complete spectra of $Rx$.
\end{definition}
Then we present the following theorem which is the base of this
paper. Let $I_d$ denotes the identity matrix of dimension $d$, and
$Diag(V)$ is a diagonal matrix with diagonal part equals $V$.
\begin{theorem}
Each autocorrelation matrix $Rx$ corresponds to a unique set of
$d+1$ nonnegative real vectors $\{S_1,S_2...,S_{d+1}\}$, where $S_k$
is the $k$-Spectrum of $Rx$ and for each $k$,
$\sum_{i=1}^{d}(S_k)i=1$. $\{S_1,S_2...,S_{d+1}\}$ can reconstruct
$Rx$ by
\begin{equation}
\label{eq-b}\\
 Rx=\sum_{i=1}^{d+1}M_i\cdot Diag(S_i)\cdot M_i^H-Id
\end{equation}
 But the inverse is not right, i.e there are some set
of $d+1$ nonnegative real vectors $\{V_1,V_2...,V_{d+1}\}$ satisfies
for each $k$, $\sum_{i=1}^{d}(V_k)i=1$, but they can't form the
complete spectra of any autocorrelation matrix.
\end{theorem}
\begin{proof}
The first part of the theorem is finished by \cite{Ivanovic}, where
we only need to switch "density matrixes" to "autocorrelation
matrixes". And it's easy to find a counterexample for the second
part. Let $V_i$ is a zero vector except the $i'th$ term which is
$1$, for $i=1,2...,d$. Then no matter how we choose $V_{d+1}$ ,
$\{V_1,V_2...,V_{d+1}\}$ can't form the spectra of some
autocorrelation matrix.
 \hfill$\Box$
\end{proof}

A trivial observation is that many different real nonnegative
vectors $\{S_1,S_2...,S_{d+1}\}$ can construct the same $Rx$
use(\ref{eq-b}). The next theorem says that it's not interesting
except for some constant global shifts to the spectrum. So as
default, in the next, we will use definition $1$ to define the
spectra of MUB. Let $One$ denotes a $d$ length vector with all term
$1$
\begin{theorem}
Nonnegative real vectors $\{S_1,S_2...,S_{d+1}\}$ and
$\{S_1',S_2'...,S_{d+1}'\}$ can construct the same $Rx$
use(\ref{eq-b}) only if for each $i=1,2,...,d+1$, there exists a
real number $u_i$, s.t $S_i=S_i'+u_i\cdot One$ .
\end{theorem}

\begin{proof}
Assume $S_k$ is the $k$-spectrum of $Rx$ by definition $1$, and:
\begin{equation}
Rx=\sum_{i=1}^{d+1}M_i\cdot Diag(S_i')\cdot M_i^H-Id
\end{equation}
For each $i\neq j$, we can check that the diagonal
part of $M_j^H\cdot M_i\cdot Diag(S_i')\cdot M_i^H \cdot M_j$ is
$u_{j,i}\cdot Id$, where $u_{j,i}$ is a real number. This finishes
the proof.
\hfill$\Box$
\end{proof}

In theorem $1$, we have shown that not all kinds of sets of positive
vectors can form a autocorrelation matrix. So what kinds of vectors
can form the complete spectra is an interesting question. Two
theorems will be presented about this subject and will be used in
next sections.
\begin{theorem}
Let $tr(Rx)=1$, and $\{S_1,S_2...,S_{d+1}\}$ form the complete
spectra of a autocorrelation matrix $Rx$, then
$\{S_1,S_2...,S_d,F\}$ also form the complete spectra of anther
autocorrelation matrix $Rx'$, where $F$ equals $\frac{1}{n}\cdot
One$.
\end{theorem}
\begin{proof}
In \cite{Ivanovic}, the author shows that if
$\{S_1,S_2...,S_{d+1}\}$ form the complete spectra of a
autocorrelation matrix $Rx$, then $Rx=\sum_{i=1}^{d+1} M_i\cdot
Diag(S_i)\cdot M_i^H-I$. He also shows that $\sum_{i=1}^{d} M_i\cdot
diag(S_i)\cdot M_i^H-\frac{n-1}{n}\cdot I$ is also a autocorrelation
matrix $Rx'$. This finishes the the proof. \hfill$\Box$ \\
\end{proof}

The next theorem is the "uncertainty principle" of the complete
spectra.
\begin{theorem}
Let $tr(Rx)=1$, $m_i$ denotes the max value of $S_i$, then:
\begin{equation}
m_j<\sqrt{2}\cdot\sqrt{1-m_i}+\frac{1}{d},i\neq j
\end{equation}
\end{theorem}
\begin{proof}
Without loss of generality, we assume $j=2$ and $i=1$. Let $Dm(Rx)$
denotes the matrix with diagonal part equaling the diagonal part of
$Rx$ and other terms equaling $0$. Let $Dv(Rx)$ denotes the vector
which equals the diagonal part of $Rx$.
\begin{eqnarray}
S_2&=&Dv(M_2^H\cdot R_x \cdot M_2)\\
&=&Dv(M_2^H\cdot Dm(R_x) \cdot M_2)+Dv(M_2^H\cdot (Rx-Dm(R_x)) \cdot
M_2)\\
&=&\frac{1}{d}\cdot One+Dv(M_2^H\cdot (Rx-Dm(R_x)) \cdot M_2)
\end{eqnarray}
$One$ denotes a $d$ length vector with each term equals $1$. Assume
$Dv(Rx)=[d_1,d_2,...d_d]^T$, and $d_1=m_1$. Because of
cauchy-schwarz inequality, we have $(Rx){i,j}\leq \sqrt{d_i\cdot
d_j}$. For matrix $M$, let $(abs(M))_{i,j}=|(M)_{i,j}|$, then:
\begin{eqnarray}
&max&(Dv(M_2^H\cdot (Rx-Dm(R_x))\cdot M_2))\\
&\leq& max(Dv(abs(M_2^H)\cdot abs((Rx-Dm(R_x)))\cdot abs(M_2)))\\
&\leq& \frac{2}{d}\sum_{i\neq j}\sqrt{d_i\cdot d_j}\\
&\leq& \frac{2}{d}(\sqrt{(d-1)\cdot d_1 \cdot
(1-d_1)}+\sqrt{(d-2)\cdot d_2 \cdot (1-d_1-d_2)}+\\
 &    & ...+\sqrt{d_{d-1}\cdot d_d})\\
&\leq& \frac{2}{d} \cdot \sqrt {1-d_1} (\sqrt{(d-1)\cdot d_1
}+\sqrt{(d-2)\cdot d_2 }+...+\sqrt{ d_{d-1} }) \\
&\leq& \frac{2}{d}\cdot \sqrt{1-d_1}\cdot \sqrt{d-1+d-2+...+1} \\
&<& \sqrt 2 \cdot \sqrt{1-d_1}
\end{eqnarray}
We get $(10)$, $(12)$ from cauchy-schwarz inequality.
\hfill$\Box$
\end{proof}

The following theorem is about the sensitivity of the equivalence
between the two representations of the autocorrelation matrixes.  We
consider the cases of random error and deterministic error. The proof
is trivial, and omitted here.
\begin{theorem}
 Let $\{S_1,S_2...,S_{d+1}\}$ is the complete spectra of $Rx$, and $E_R$
is a error matrix , $E_{Si}$ is a error vector. Assume that $Rx+E_R$
is also positive and
$\{S_1+E_{S1},S_2+E_{S2}...,S_{d+1}+E_{S(d+1)}\}$ is also the
complete spectra of a autocorrelation matrix.

 (i) If $E_R$ is deterministic error matrix of $Rx$ satisfies $|E_R|_{\oe}
<\epsilon$,  then the complete spectra of $E_R+Rx$ is
$\{S_1+E_{S1},S_2+E_{S2}...,S_{d+1}+E_{S(d+1)}\}$, satisfies
$|E_{Si}|_{\oe}<d\cdot \epsilon$, for $i=1,2,...d+1$.

(ii) If $E_R$ is random error matrix of $Rx$ satisfies each term of
$E_R$ are independent, $E((E_R)_{i,j})=0$, and
$E(((E_R)_{i,j}))^2<\epsilon$ for all $i,j = 1,2,...d$. Then the
complete spectra of $E_R+Rx$ is
$\{S_1+E_{S1},S_2+E_{S2}...,S_{d+1}+E_{S(d+1)}\}$, satisfies
$E((E_{Si})_{j})=0, E(((E_{Si})_{j}))^2<\epsilon$, for
$i,j=1,2,...d+1$.

(iii)If for each $i$, $E_{Si}$ is deterministic error vector of
$S_i$ satisfies $|E_{Si}|_{\oe} <\epsilon$, then
$\{S_1+E_{S1},S_2+E_{S2}...,S_{d+1}+E_{S(d+1)}\}$ form the complete
spectra of $E_R+Rx$ , where $|E_R|_{\oe}<n\cdot \epsilon$.

(iv)If for each $i$, $E_{Si}$ is random error vector of $S_i$
satisfies $(E_{Si})_j$ are independent for each $i = 1,2...d+1$ and
$j=1,2,...,d$, and $E((E_{Si})_{j})=0,
E(((E_{Si})_{j}))^2<\epsilon$. Then
$\{S_1+E_{S1},S_2+E_{S2}...,S_{d+1}+E_{S(d+1)}\}$ form the complete
spectra of $E_R+Rx$ , where $E((E_R)_{i,j})=0,
E(((E_R)_{i,j}))^2<\epsilon$.
\end{theorem}
\begin{center}
\vskip 1em
 III. THINGS BECOMES CLEAR WHEN MUB COMES
 \vskip 1em
\end{center}
\vskip 1em \noindent{\small\emph{A. the Generalization of the
Definition of Stationary}

 Stationary random signal is easy in the sense that we can
apply Fourier spectra analysis. But things become much harder when
the signal is nonstationary. In this subsection, the definition of
stationary random vector is extremely extended by MUB. This
extension is serious, because it concerns which domains we should
concern to do complete signal analysis. In this subsection, $X$,
$X'$ are two random complex vectors, $Rx$ and $Rx'$ are
autocorrelation matrixes of $X$, $X'$, and$\{S_1,S_2...,S_{d+1}\}$
and$\{S_1',S_2'...,S_{d+1}'\}$are the complete spectra of $Rx$,
$Rx'$. $F$ also equals $1/n\cdot One$.
\begin{definition}
X is $[i_1,i_2,...i_k]$-stationary if and only if
$S_{i_1}=S_{i_2}=...=S_{i_k}=tr(Rx)\cdot F$
\end{definition}

\begin{definition}
X' is $[i_1,i_2,...i_k]$-stabilizer of X  if and only if
$S_{i_1}'=S_{i_2}'=...=S_{i_k}'=tr(Rx)\cdot F$, and $S_j'=S_j$ for
each $j\not\in \{i_1,i_2...i_k\}$
\end{definition}

{\em Proposition I}   Every $X$ can have all kinds of stabilizer
because of theorem $3$.

One should notice that "stabilization" is an information lossing
process. And $[i_1,i_2,...i_k]$ stabilizer of $X$ will left the
information of $j$ spectrum of $X$ unchanged, when
$j\not\in\{i_1,i_2,...,i_k\}$. However, it will be shown that this
process can protect the information of some spectra. And a general
way to stabilize signals will be presented.

There are two interesting propositions which concerns some traditional
important properties of random vector.

{\em Proposition II} If $M_2$ is the Fourier  base, X is
"stationary"(in original sense) if and only if X is $[1,3,4...,d+1]$
stationary.

{\em Proposition III}  X is "white noise"(in original sense) if and
only if X is $[1,2,...,d+1]$ stationary.

\noindent{\small\emph{B.Correlation and Independent}

 \indent \indent In this part, some relationships between normal related random sources
 and independent random sources will be presented . Let $m(Rx)$ denotes the minimum eigenvalue of $Rx$, $m_s(i)$
denote the minimum term of $S_i$. We first give the main theorem of
this subsection.

\begin{theorem}
If $Rx$ is a autocorrelation matrix with complete spectra
$\{S_1,S_2...,S_{d+1}\}$, and $tr(Rx)=1$. If :
\begin{equation}
\label{eq-200}\\
  m_s(i)\geq \frac{1}{n+1},i=1,2,...,d+1
\end{equation}
, then we can construct a complex random vector with autocorrelation
matrix $Rx$ by $d\cdot (d+1)$  independent random variables.
\end{theorem}

\begin{proof}
From \cite{Ivanovic}, Let ,we have:
\begin{eqnarray}
Rx&=&\sum_{i=1}^{d+1} M_i\cdot diag(S_i)\cdot M_i^H-I\\
Rx&=&\sum_{i=1}^{d+1} M_i\cdot (diag(S_i)-\frac{1}{n+1}\cdot I)\cdot
M_i^H
 \label{eq-1}
\end{eqnarray}
If (\ref{eq-200}) holds, We can construct $d+1$ random vectors
$\{Y_1,Y_2,...Y_{d+1}\}$, satisfies $\{(Y_i)_j, i=1,2,...d+1,
j=1,2,...d\}$ are  independent random variables.
For each i,j of available values, $(Y_i)_j$ satisfies:\\
\begin{eqnarray}
E((Y_i)_j)&=&0\\
E((Y_i)_j^2)&=&(S_i)_j-\frac{1}{n+1} \\
\end{eqnarray}
Let:
\begin{equation}
\label{eq-4}
 X=\sum_{i=1}^{d+1} M_i\cdot Y_i
\end{equation}
Then autocorrelation matrix of $X$ is $Rx$.\\
\hfill$\Box$
 \end{proof}

{\em Remark I}  If $m(R)\geq \frac{1}{n+1}$£¬then
(\ref{eq-200})holds.

{\em Remark II}  With theorem 2, one can shows that (\ref{eq-200})
can be replaced by a weaker one: sum of $m_s(i)$ is no less than
$1$. But still a lot of autocorrelation matrixes fail to satisfy it.

Remark II seems a strong constraint, but in the next subsection , we
will see that in some place it can be overcome easily, while in
others, it will lead some natural results.

For convenience, we define:
 \begin{definition}
 $X$ is a $k$-domain random vector if and only if $X=M_k\cdot Y$,
 where $Y$ is a $d$ dimension random vector satisfies $E(Y)=0$, and
 the terms of $Y$ are independent.
 \end{definition}

It will be shown that the alternation between $X$ and $\{Y_i,
i=1,2...,d+1\}$ is very useful in various ares. Generally speaking,
$X+N$ can be viewed as composition of independent random vectors
from different domains, where $N$ is "white noise" with $E(N\cdot
N^H)=tr(R_x)\cdot I$. It should be noticed that $N$ and $k$-domain
random vector has nothing to do with the $r$-spectrum except for a
global incensement/decrement if $r\neq k$. Or we can think of $X$ is
a composition of independent random variables from different domains
with a denoise procedure in the end. This suggests that we can just
treat signals as a set different independent signals from different
domains, and energy distribution on each domain won't change after
the composition except for a global constant shift. In other words,
$i$-spectrum has nothing to say about the energy distribution of the
$j$-spectrum when $i\neq j$.

\noindent{\small\emph{C. Linear Random Operator and Some Special
Kinds of Filters}

In this subsection, we will do something in the taste of signal
processing . The reader will see that linear operators/filters for
random vectors will be demonstrated clearly with MUB.
And we can judge whether a filter is good in the sense that it only
do what it should and left other parts untouched.

A general formulation of linear random operators is a good start
point to study complete MUB analysis for operators. Reminding that a
random variable is "white noise" only if it's independent with any
other random variables in this paper.
\begin{definition}
 P is a random linear operator for $d$ dimension complex random vector, if:
 \begin{equation}
 \label{eq-5}
 P(X)=T\cdot X
 \end{equation}
 Where T is a random $d\cdot d$ matrix . And for each
 subset $Sub_x$ of $\{(X)_1,(X)_2,...,(X)_d\}$, each subset $Sub_T$ of $\{(T)_{i,j},i,j=1,2,...,d\}$
 satisfies:
 \begin{equation}
 \label{eq-6}
 Pr_\{Sub_X,Sub_T\}=Pr_{Sub_X}\cdot Pr_{Sub_T}
 \end{equation}
 \end{definition}
There are some propositions for P, which are trivial but important.

{\em Proposition I}  For random vectors $X$ and $X'$, if $Rx=Rx'$,
then $R_{P(X)}=R_{P(X')}$.

{\em Proposition II}  For random vectors $X$ and $X'$, if $E(X'\cdot
X^H)=0$, then $R_{P(X+X')}=R_{P(X')}+R_{P(X)}$.

Then we the main theorem of this subsection:
\begin{theorem}
For random vector $X$ with $tr(R_x)=1$, $\{S_1,S_2...,S_{d+1}\}$ are the complete
spectra of $Rx$, $\{S_{p1},S_{p2}...,S_{p(d+1)}\}$ are the complete
spectra of $R_{P(x)}$. There exist $d+1$  dimension $d*(d^2+d)$
deterministic real matrixes $\{D_1,D_2,...D_{d+1}\}$, such that for
$i=1,2,...d+1$,
\begin{equation}
\label{eq-7} \\
S_{pi}=D_i\cdot [(S_1-\frac{1}{d+1}\cdot
One)^T,(S_2-\frac{1}{d+1}\cdot
One)^T,...,(S_{d+1}-\frac{1}{d+1}\cdot One)^T]^T
\end{equation}
$One$ is a $d$ length vector with each term equals $1$
\end{theorem}
\begin{proof}
First assumes that $m(Rx)\geq \frac{1}{n+1}$, then from theorem $5$,
There exist $d+1$ random vectors $\{Y_1,Y_2,...Y_{d+1}\}$, satisfies
$\{(Y_i)_j, i=1,2,...d+1, j=1,2,...d\}$ are independent random
variables, $E(Y_i)=0$, and $E({(Y_i)_j}^2)=(S_i)_j-\frac{1}{d+1}$.
Let
\begin{equation}
X'=\sum_{i=1}^{d+1} M_i\cdot Y_i
\end{equation}
Then $Rx=Rx'$. From proposition I,II,
\begin{equation}
\label{eq-8}\\
R_{P(X)}=R_{P(X')}=\sum_{i=1}^{d+1}\sum_{j=1}^{d} R_{P(M_i\cdot
Z_{j})}\cdot ((S_i)_j-\frac{1}{d+1})
\end{equation}

$Z_{j}$ is the random vector satisfies $E(Z_{j})=0$, and
$E((Z_{j})_i^2)=\delta_{ij}$. (\ref{eq-8}) has already shown the
properties of $P$ can be determined by some deterministic matrixes,
but we need to go further.

 For each $R_{P(M_i\cdot Z_{j})}$, the $k$
spectrum is $S_k^{i,j}$, so the $k$ spectrum of $R_{P(X)}$ is:
\begin{equation}
\label{eq-9}
\\
S_{pk}=\sum_{i=1}^{d+1}\sum_{j=1}^{d} S_k^{i,j} \cdot
((S_i)_j-\frac{1}{d+1})
\end{equation}
So there exists $\{D_1,D_2,...D_{d+1}\}$ satisfies (\ref{eq-7}).

The second part of the proof will deal with the constraint
"$m(Rx)\geq \frac{1}{n+1}$". Let $X_n=\frac{1}{\sqrt{d+1}} \cdot
(X+N)$, $N$ is d length "white noise" with $E(N\cdot N^H)=I$. Now
$m(Rx_n)\geq \frac{1}{n+1}$ holds. Let $L_{ONE}$ denotes a $d^2+d$
length vector with each term is $1$. Let $S_{pk}^{(X_n)}$ denotes
the $k$-spectrum of $R_{P(X_n)}$, and $S_{pk}^{(N)}$ corresponds to
the $k$-spectrum of $R_{P(N)}$. Then
\begin{eqnarray}
S_{pk}^{(X_n)}&=&\frac{1}{d+1}\cdot D_k \cdot
[S_1^T,S_2^T,...,S_{d+1}^T]^T\\
S_{pk}^{(X_n)}&=&\frac{1}{d+1}\cdot (S_{pk}+S_{pk}^{(N)})\\
S_{pk}^{(N)}&=& \frac{1}{d+1}\cdot D_k \cdot L_{ONE}^T   \\
\end{eqnarray}
From above equations, (\ref{eq-7}) holds.
\hfill$\Box$
\end{proof}

$\{D_1,D_2,...D_{d+1}\}$ shows some basic property of P. For
example, Let $I_d$ denote the identity $d*d$ matrix, $1_d$ is a
$d*d$ matrix with each term equals $1$, if there are some real numbers $\mu_i,i=1,2,...,d+1$ such that:
\begin{equation}
D_k=[\mu_1\cdot 1_d, \mu_2\cdot 1_d,...\mu_{d+1}\cdot 1_d]
\end{equation}
Then the output of P will be $[k]$ stationary. If
\begin{equation}
D_k=[\mu_1\cdot 1_d, \mu_2\cdot 1_d,...,\mu_{i-1}\cdot
1_d,\mu_{i}\cdot I_d, \mu_{i+1}\cdot 1_d,...,\mu_{d+1}\cdot 1_d]
\end{equation}
Then P actually switch the $i$-spectrum of $X$ to the $k$-spectrum
of $P(X)$ with a global constant increment/decrement .

If a filter only want to do something about the $j$-spectrum and
keep the information of other spectra unchanged, then it should try
to satisfy :
\begin{equation}
D_k=[\mu_1\cdot 1_d, \mu_2\cdot 1_d,...,\mu_{k-1}\cdot
1_d,\mu_{k}\cdot I_d, \mu_{k+1}\cdot 1_d,...,\mu_{d+1}\cdot
1_d],k\neq j
\end{equation}

For example, let the matrix $T$ of operator $P$ be a deterministic
matrix with the form $[V,S_1(V),S_2(V),...,S_{d-1}(V)]^H$, while $V$
is a length $d$ vector and $S_i(V)$ means the vector which left ring shift $V$ $i$
times. This kind of $P$ is well studied(such as Winner Filter). We
could also say that kind of $P$ is a good $2$-spectrum filter if the
in put signals $X$ are $[1,3,4...,d+1]$ stationary, i.e
$E((X)_i\cdot (X)_j^H)=F(j-i)$, where $F$ satisfies $F(k)=F(-k)^H$.

So it's a very interesting question that what kinds of
$\{D_1,D_2...,D_{d+1}\}$ corresponds to a physical realizable random
operator, but this paper can't answer it .

 The second part of this subsection will focus on some special
kinds of operators.

\begin{theorem}
For $k=1,2,...d$, there exist a operator $ P_{i_1,i_2,...i_k}$, such that for
any input $X$, it will output $X_{i_1,i_2,...i_k}+N$. $X_{i_1,i_2,...i_k}$
is the $[i_1,i_2,...i_k]$ stabilizer of $X$ and $N$ is "white noise"
with $E(N\cdot N^H)=tr(R_X) \cdot \frac{d-k}{d}\cdot I_d$.
\end{theorem}

\begin{proof}
Let $Ys^{(j)}=diag(y_1^{(j)},y_2^{(j)},...,y_{d}^{(j)})$, for each
$j \not\in \{i_1,i_2,...i_k\}$, satisfies that $Y=\{y_i^{(j)},j
\not\in \{i_1,i_2,...i_k\}, i=1,2...,d\}$ are independent random
variables, $E(y_i^{(j)})=0$ and $E((y_i^{(j)})^2)=1$. $Y$ are
independent to $X$. Let $X'$ is the output of $P_{i_1,i_2,...i_k}$
with input $X$ . $X'$ is constructed from the following equations.

\begin{equation}
X'=\sum_{j\not\in \{i_1,i_2,...i_k\}} M_j\cdot Ys^{(j)}\cdot
M_j^H\cdot X
\end{equation}
Then the $t$ spectrum $S_t'$ of $X'$ is:
\begin{eqnarray}
S_t'&=&S_t+tr(Rx)\cdot\frac{d-k}{d}\cdot One,t\not\in \{i_1,i_2,...i_k\}\\
S_t'&=&tr(Rx)\cdot\frac{d-k+1}{d}\cdot One,else
\end{eqnarray}
This finishes the proof.
 \hfill$\Box$
 \end{proof}

It's natural to see that more precise stabilization needs more
compensations on "white noise". It is also very interesting to study
how to lowerbound the "noise compensation" of "stabilization".

Based on above techniques, we can also construct a special kind of
filter mentioned above, the one that only works on designated
spectra. For example, think of the case that we only want to do
something in the Fourier domain. Based on above technique, we can
first choose a suitable value for $E((y_i^{(2)})^2)$, for
$i=1,2,...,d$. Then we output $M_2\cdot Ys^{(2)}\cdot M_2^H\cdot X
+X$. This filter only change the Fourier spectrum with compensation
on "white noise".
\begin{center}
\vskip 1em
 IV. MUB TRANSFORMATION FOR DETERMINISTIC VECTORS
 \vskip 1em
\end{center}
\vskip 1em From now on, X becomes a deterministic vector of $d$
dimension complex linear space, and $X$'s $k$-spectrum is denoted by
$S_k=M_k^H\cdot X$. When $k$ is a odd prime number, and $M_1$
denotes $I_d$, $M_k$ with $k>1$ can be constructed by the formulae
\cite{newproof}:
\begin{equation}
\label{eq-11}
 (M_k)_{j,r}=W^{r\cdot j+\frac{(k-2)\cdot(j^2-j)}{2}}
\end{equation}
$i$ is the square root of $-1$, and $W=e^{\frac{2\pi i}{d}}$. A trivial observation is
 that $M_2$ is the discrete fourier matrix. The
following theorem says that for each $k$, $k$-spectrum $S_k$ of $X$
can be found from X nearly as fast as the $2$-spectrum which could
use FFT.
\begin{theorem}
If for any $k=1,2,...d+1$, $M_k$ is constructed from (\ref{eq-11}),
$T_k$ denotes the time needed to compute $S_k$ from $X$, $T_k'$
denotes the time needed to compute $X$ from $S_k$, then $T_k\leq
T_2+d\cdot T_m $ and $T_k' \leq T_2'+d\cdot T_m$, where $T_m$ is the
time need complex multiplication.
\end{theorem}
\begin{proof}
Let $H_k=diag(h_1^{(k)},h_2^{(k)},...h_d^{(k)})$, where
$h_j^{(k)}=W^{\frac{k\cdot(j^2-j)}{2}}$. Then :
\begin{eqnarray}
M_k&=&H_k\cdot M_2\\
X&=&M_k \cdot S_k=H_k\cdot M_2 \cdot S_k \\
M_k^H&=&M_2^H\cdot H_k^H\\
S_k&=&M_k^H\cdot X=M_2^H\cdot H_k^H \cdot X
\end{eqnarray}
This finishes the proof.
 \hfill$\Box$
 \end{proof}

Similar to DFT, there are also some interesting properties for the
MUB spectra of $X$.
\begin{theorem}
For a normalized complex vector $X$,  let $m_i=|S_j|_{\oe}$ ,then
the following holds:
\begin{equation}
 m_j<\frac{1}{\sqrt d}\cdot m_i+\sqrt {1-m_i^2},j\neq i
 \end{equation}
\end{theorem}
\begin{proof}
With out loss of generality, assumes $|(S_i)_1|=m_i$.If $j\neq i$,
we have:
\begin{eqnarray}
m_j&\leq& \frac {1}{\sqrt d}\cdot \sum_{k=1}^{d}|(S_i)_k|\\
&\leq &\frac{1}{\sqrt d}\cdot m_i+\frac {1}{\sqrt d}\cdot
\sum_{k=2}^{d}|(S_i)_k|\\
\label{eq-13} &\leq& \frac{1}{\sqrt d}\cdot m_i+\frac{\sqrt
{d-1}}{\sqrt d}\cdot \sqrt{1-m_i^2}\\
&<& \frac{1}{\sqrt d}\cdot m_i+\sqrt {1-m_i^2}
\end{eqnarray}

It's easy to see that(\ref {eq-13}) comes from cauchy-schwarz
inequality. \hfill$\Box$ \\ \end{proof}

The above theorem can be thought of the generalization of original
"uncertainty principle" for deterministic vectors. While the next
theorem is a positive result about the MUB spectra.

\begin{theorem}
For any normalized complex vector $X$ , there exists $k\in [1,d+1]$
,satisfies $|S_k|_{\oe}>\frac {1}{\sqrt d}$
\end{theorem}
\begin{proof}
Let $V_x$ is the vector which contradict the theorem, construct a
$d*d$ matrix $A=[Vx,Vx,...Vx]$, and $B=A\cdot A^H$. All the $d*d$
matrix forms a $d^2$ linear space, and $B$ is not in the subspace of all
diagonal matrixes. It's easy to check $B$ is orthogonal to all the
non-diagonal matrixes constructed in theorem $3.4$ of
\cite{newproof}, which implies there are at least $d^2+1$ orthogonal
bases for $d*d$ matrixes. \hfill$\Box$ \\ \end{proof}

If $d$ is prime and MUB are constructed from (\ref{eq-11}),
numerical analysis suggests that complete MUB spectra of $X$ also
have many interesting properties similar to the spectrum of DFT,
such as symmetry of $X$ will leads interesting symmetries for all
MUB spectra , and ring shifts of $X$ also cause some shifts of all
MUB spectra in the sense of absolute values.

\begin{center}
\vskip 1em
 V. ENCODE INFORMATION INTO SECOND MOMENTS
 \vskip 1em
\end{center}
\vskip 1em
 The main application of above results is a simple digital
communication protocol which can significantly increase the number
of users who can use the channel simultaneously and worst case
bounded. Although the theoretical protocol is far from practice, it
has provided some deep insights about how to encode information into
the second moments between random variables. Based on the results of
Section $III$, we will introduce some interesting alternations of
the model which suggest we can do many things based on such model.

First we assume that $\{A_1,A_2,...,A_n\}$ are all nodes who want to
communicate with others. There is only a public discrete complex
channel $C$ for them to communicate. In the first half of each time
interval, $C$ collect a complex message $Mes_i$ from $A_i$,  sum
$Mes_i$ all to $Mes$, and send $Mes$ to each $A_i$ in the second
half of the interval.

We assume for every $d$ intervals, $C$ will give an synchronous
impulse to each $A_i$ which can be distinguished from messages. The
abilities of $A_i$ are constraint, they can't count the impulses
from the start. Actually, the impulses can be thought as the frame
synchronous signal of the channel, and this model is the base for
multiple access digital communication\cite{DigitalCom}, such as
TDMA/FDMA\cite{TDMA,tdma2}. Since the number of digital
communication users grows fast, scientists invent many advance
techniques to handle large size systems, such as the one which
combine TDMA and FDMA together \cite{Trans}. In this part, we
present a easier way to increase the number of users. We will also
study some interesting alternations of $C$ later.

We define the protocol is $(n,d,m)$ worst case good on $g$ if there
exists a function $g$ , such that from the start time when $A_i$
wants to send a $k$ bit message , there only needs $g(n,d,k)$ time
intervals to make sure that the probability that $A_j$ can get right
information from $A_i$ is at least 2/3, for each $j\neq i$.

It's easy to see that when $n=O(d)$, the protocol is good because of
TDMA or FDMA. If we have more users, we can use the idea of
arithmetic coding\cite{Arthcode}, but it's hard to be applied to
large system because $t$ times users needs $2^t$ times power cost
for some users.  Now we introduce a easily applied protocol which
can square the number of users. Actually, it's a protocol which is
worst case good on function $g=O(d^5\cdot lg(k))\cdot k$, and
requests at most $d^3$ times of power cost. Numerical analysis
suggests that when $d=127$ and $n=254$, then within $100*127$ time
intervals, $A_j$ can get the right information from $A_i$ with high
probability.

The idea of the protocol is very simple. We first assume the
messages are all positive real numbers . Then we assign each $A_i$ a
special range, such as time range or frequency range. When $A_i$
want to send some messages, he first flips coins and gives some
random signs to the messages. After that, $A_i$ sends the message
which are coded in his designated range. The key is that if $X$ is a
composition of random vectors $\{V_k,k\in [1,d+1]\}$ from different
domains(see definition $6$ ), then the energy distribution of $X$ on
domain $k$ is the same as $V_k$ except for a global constant shift.
For example, we will give the protocol when $d=4$, $n=10$.

Assume for $i=1,2,..,10$, the the messages of $A_i$ are two real
numbers $Mo_1^{(i)}$ and $Mo_2^{(i)}$ of $[0,1]$, and he wants to
tell others which one is lager. Let $M_1$, $M_2$,..,$M_5$ are the
MUB of dimension $4$. For $A_i$, we assign $M_{[\frac {i+1}{2}]}$ to
him, where $[x]$ means the integer part of $x$. To communicate,
$A_i$ first compute $Mes^{(i)}=\sqrt{Mo^{(i)}}$. Then for each
round, $A_i$ flips two coins, and change the sign of $Mes_j^{(i)}$
if he got "heads" at the $j'th$ flipping, $j=1,2$. Then he computes
$Vi$ by
\begin{eqnarray}
V_i&=&M_{[\frac {i+1}{2}]}\cdot[Mes_1^{(i)},Mes_2^{(i)},0,0]^T,i=odd\\
V_i&=&M_{[\frac {i+1}{2}]}\cdot[0,0,Mes_1^{(i)},Mes_2^{(i)}]^T,i=even\\
\end{eqnarray}

When each synchronous impulse comes, $A_i$ send $V_i$ one by one to
the channel $C$. For $A_j$, he receives signals one by one from $C$.
Assume the signals in this round form a $d$ length complex vector
$X$. For $A_j$, he needs to keep $|(M_1^H\cdot X)_1|^2$ and
$|(M_1^H\cdot X)_2|^2$ for the information of $A_1$, and keep the
data for other $A_i$ in a similar way, $i\neq j$. Then after a
$1000$ rounds, $A_j$ can tell whether $Mo_1^{i}$ is larger than
$Mo_2^{i}$ correctly with high probability.

In general case, we count the rounds needed for $A_j$ theoretically.
We assume $A_j$ wants to recover $E(|(M_1'\cdot X)_1|^2)$. In the
communication , some users of domain $k$, $k\neq 1$, may start/stop
to sent signals to $C$. It doesn't matter, because for $A_j$, they
are global looked same noise and won't effect the relation between
$E(|(M_1'\cdot X)_1|^2)$ and $E(|(M_1'\cdot X)_2|^2)$. So we only
consider the case when the total energy of all domains except $1$ is
upper bounded by $K$.

We know that $X$ is constructed from independent random variables
$\{n_{j}^{(i)}\}$ from different domains, where $n_{j}^{(i)}$
denotes the $j'th$ random variable from domain $i$. Because
$n_{j}^{(1)}$ won't effect $|(M_1'\cdot X)_1|^2$ for $j=2,3,...,d$,
we compute the standard deviation $\sigma (|(M_1'\cdot X)_1|^2)$ by:
\begin{eqnarray}
\sigma ^2(|(M_1'\cdot X)_1|^2)&\leq& O( E(|(M_1'\cdot X)_1|^4))\\
&\leq& O (\sum_{j_1,j_2,i_1\neq 1,i_2\neq 1}\frac {1}{d^2}\cdot
E((n_{j_1}^{(i_1)})^2)\cdot E((n_{j_2}^{(i_2)})^2))\\
&\leq& O(\frac {K^2}{d^2})
\end{eqnarray}

$K$ denotes the total energy from all the domains except $1$. Let
$M(|(M_1'\cdot X)_1|^2)$ denotes the mean value of $|(M_1'\cdot
X)_1|^2$ in $r$ rounds. Using chernoff bound, we conclude when
$r=O(K^2\cdot lg(k))$, we have:
\begin{equation}
Pro(M(|(M_1'\cdot X)_1|^2)-E(|(M_1'\cdot
X)_1|^2)>O(\frac{1}{d}))\leq O(\frac{1}{k})
\end{equation}

So the probability that all $k$ bits are correct is more than a
constant positive value. In the worst case, when $K=d^2$, we need
$O(d^5\cdot lg(k)\cdot k)$ time intervals to make sure $A_j$ can
receive the right information from $A_i$ with probability larger
than $\frac{2}{3}$.

Next we consider the error from quantification. It's easy to check
that when the error of $(X)_i$ is less than $\epsilon$, for
$i=1,2,..,d$, then error of $|(M_1'\cdot X)_1|^2$ is less than
$O(d^2\cdot\epsilon) $. So if $\epsilon<\frac{1}{d^3}$, the mean
error of $|(M_1'\cdot X)_1|^2$ from quantification will be less than
$O(\frac{1}{d})$. For each $A_i$, he need to quantify the the
signal(sent or received) to $O(d^3)$ discrete magnitude values and
$O(d^3)$ discrete phase values to satisfy $\epsilon<\frac{1}{d^3}$.

If only time/frequence resources are allowed, the protocol is just
"TDMA/FDMA". When the case that more than one domains from MUB are
used, we must bounds the total energy of each domain because it's
the "noise" of other domains. There is a trade off in this model,
when more users work simultaneously, more noise comes, so more
rounds are needed. But the rounds needed for $A_i$ will be upper
bounded by a function which only concerns $n,d,k$. Although each
user can choose any time to start or end a communication process, a
better choice is to choose a time when the energy of his designated
range is low, which may bring a average optimization to the whole
system. So when the frequency resource is in shortage, and it's not
suitable to apply some advanced techniques to the system, it seems a
reasonable way to allocate resources to great numerous of users, for
the reason that it's adaptable, analyzable, and worst case bounded.

Actually, traditional protocols such as "TDMA" are based on the
first moments of the signal, while the highlight of our protocol is
that it can fully utilize the information of the second moments of
the signals.

Next we'll focus on some special kinds of channels/filters  $C$
based on subsection $C$ of Section $III$. We study how can $C$
process the information of each $A_i$.

First, when $C$ has "white noise" $N$, then $N$ effects all the
users equivalently as "white noise".

Second, if $C$ can be described by some deterministic matrixes
$\{D_1,D_2,...,D_{d+1}\}$ (See theorem $7$), $C$ will do what we
claimed in the part following theorem $7$. So we can choose the
domains that have nice properties to realize the protocol.

Third, follow the idea of theorem $8$, $C$ can do something special
to $A_i$. Such as $C$ can change the information of $A_i$ without
effect others except for some global looked same "noise". Actually,
$C$ can stabilizes the range designated to $A_i$ so nobody can know
the information from $A_i$.

Compared to traditional protocol, such as "TDMA", $C$ can almost do
all the job the channel $C_T$ of "TDMA" can do. Even more, $C$ also
can do things $C_T$ can't do, such as $C$ can switch the information
from different domains. However, almost every special thing $C$ can
do will bring "noise". So the question raised before that "what
kinds of $\{D_1,D_2,...,D_{d+1}\}$ correspond to a physical
realizable filter" becomes important.
\begin{center}
\vskip 1em
 VI. DISCRETE SIGNAL ANALYSIS WITH MUB
 \vskip 1em
\end{center}
\vskip 1em
 In subsection $B$ of Section $III$, the traditional definition of
"Stationary" is extremely extended by MUB. And subsection $C$ of
Section $III$ suggests the spectra which are far from stationary
must implies some nontrivial information in their domains. Actually,
if we treat discrete signals as a composition of independent signals
from different domains, then spectra analysis in any domain has its
own meaning: the $k$-spectrum uniquely describe the energy
distribution of the $k$-domain random vector except for a global
constant shift. So Fourier spectrum analysis also makes sense when
the signal is nonstationary.

Subsection $C$ of Section $III$ gives some ideas about how to
construct filters to process statistic signals. These filters are
different from traditional ones in the sense that they must concerns
all the spectra which we are interested in.

Next, for signal detection,  we give a definition regarded to how to
judge whether a signal is meaningful.
\begin{definition}
The $k$-spectrum entropy of $X$ is defined
$E_k(X)=\sum_{i=1}^{d}(-lg(\frac {(S_k)_i}{tr(Rx)}))$, the complete
entropy of $X$ is defined $E_c(X)=\sum_{j=1}^{d+1}E_j(X)$.
\end{definition}

So meaningful signals should has $Ec$ less than $d\cdot(d+1)\cdot
lg(d)$. And a signal with $E_2$ much less than $d\cdot lg(d)$ must
implies some important information in the Fourier domain, no matter
whether the signal is stationary or not.

However, the most important thing left in this part is how to
justify the physical meanings of each base. This paper failed to
achieve it. Unlike the Fourier base, for other bases from MUB, it's
looks impossible to correspond them to continuous functional
transformations when we only use the the construction when $d$ is
prime. Roughly speaking, the MUB spectra based on the constructions
when $d$ is prime is very sensitive to $d$. For instance, when a
vector has only a single point in the $k>2$-spectrum for dimension
$d$, then it will change a lot when consider the $d'>d$ dimension's
$k$-spectrum, and the larger $k$ , the more change. Whatever, the
paper suggests that if the physical meaning of a base (such as the
Fourier base) has been found, then do spectra analysis of such base
will always make sense.

To achieve to goal, we need the efforts from various areas. Such as
we need scientists from the areas of signal processing, physics, and
bioinformatics to find some physical meanings of spectra which are
definitely different from frequency. And we also need mathematicians
to tell us how to construct MUB which have as many good properties
as possible (such as the Fourier bases).
\begin{center}
\vskip 1em
 VII. DIMENSIONALITY REDUCTION WITH MUB
 \vskip 1em
\end{center}
\vskip 1em
 For information lossy data compression such as
dimensionality reduction, sometimes it's hard to have a good
compression ratio when few prior knowledge is known, and things
become even worse when the data looks like "white
noise"\cite{PCA,PCA2}.  In this section, we claim that Mutually
Unbiased Based can do the looks impossible job in some sense.

In the following, compress $X$ with MUB means choosing a subset
$Sub_m$ of all MUB bases, and find a optimal MUB spectrum of $Sub_m$
to express $X$, which need only $lg(d)$ bits to denote which base
has been chosen.

 Theorem $8$ is a technical reason that engineers
can choose any unbiased base to do data transformation, theorem $9$
suggest that not all spectra can look good, and theorem $10$ makes
sure that the worst case won't happen when whole MUB spectra are
considered. Next, we will do something different.

 $Sp$ denotes the unit sphere of $d$ dimension complex
linear space, i.e $Sp=\{V|<V,V>|_2 =1, V\in C^d\}$. For any subset
$Sub_{Sp}$ of $Sp$, $V(Sub_{Sp})$ denotes its standard volume metric
of $d$ dimension complex sphere \cite{Geometry}.

A normalized uniform random vector is a good start point to analysis
the case when no prior knowledge is known.
\begin{definition}
$X$ is a normalized uniform random vector if and only if :
\begin{equation}
Pr(X\in Sub_{Sp})=\frac {V(Sub_{Sp})}{V(Sp)}
\end{equation}
\end{definition}

In the following, compressing $X$ with $k$ normalized unitary
matrixes $\{B_1,B_2,...,B_k\}$ means choosing a optimal spectrum of
these bases to express $X$, which needs only $lg(k)$ bits to denote
which base has been chosen. First we assume $k\leq d+1$ bases from
MUB are chosen, and the max absolute value of $X$'s $i$-spectrum is
$m_i$. Then we arbitrarily choose $k$ unitary normalized matrixes
$U_1,U_2,...,U_k$ , and let $u_i=|U_i\cdot X|\oe$. We often wants to
find some spectrum with large entry to express $X$. The following
theorem justifies that the bases from MUB will do better than any
$\{U_i\}$ locally .
\begin{theorem}
When $X$ is a normalized uniform random vector, then:
\begin{equation}
Pr(max(m_1,m_2,...m_k)\geq \sqrt{\frac {d}{2d+1-2\sqrt{d}}})\geq
Pr(max(u_1,u_2,...u_k)\geq \sqrt{\frac {d}{2d+1-2\sqrt{d}}})
\end{equation}
\end{theorem}
\begin{proof}
First, a lemma will be shown :
\begin{lemma}
If $V_1$, $V_2$ are two normalized $d$ length complex vectors
satisfies $|<V_1,V_2>|\leq \frac {1}{\sqrt d}$. Then for any
normalized vector $V$, if $|<V,V_1>|=|<V,V_2>|=C$, we have :
\begin{equation}
C\leq \sqrt{\frac {d}{2d+1-2\sqrt{d}}}
\end{equation}
\end{lemma}
\begin{proof}
There exist some vector normalized $W$, $|<W,V_1>|=0$, and
$V=e^{i\theta_1}\cdot C\cdot V_1 +\sqrt{1-C^2}W$, $i$ is the square
root of $-1$.Then we have:
\begin{eqnarray}
C=|<V,V_2>|&=&|\frac{1}{\sqrt{d}}\cdot C\cdot
e^{i\theta_1}+\sqrt{1-C^2}\cdot <W,V_2>|\\
&\leq& C\cdot \frac{1}{\sqrt d}+\sqrt{1-C^2}
\end{eqnarray}
From above inequality, we can prove the lemma. \hfill$\Box$ \\
\end{proof}

For any vector $V_0$ and constant $C$, let :
\begin{equation}
De(V_0,C) =\{V:|V|_2=1, | <V,V_0>|>C, V \in C^d\}
\end{equation}
If $C=\sqrt{\frac {d}{2d+1-2\sqrt{d}}}$, and $V_i$,$V_j$ are any two
unequal vectors from MUB, then:
\begin{equation}
De(V_i,C) \cap De(V_j,C)=\emptyset
\end{equation}
So
\begin{eqnarray}
Pr(max(m_1,m_2,...m_k)\geq C)&\geq& k\cdot
d\cdot \frac {V(De(V_0,C))}{V(Sp)} \\
&\geq& Pr(max(u_1,u_2,...u_k)
\end{eqnarray}
\hfill$\Box$ \\ \end{proof}

 {\em Remark I} When $d$ goes to
infinity, $\sqrt{\frac {d}{2d+1-2\sqrt{d}}}$ limits to $\sqrt{2}$.

{\em Remark II} When $d$ goes to infinity, $d\cdot (d+1)\cdot \frac
{V(De(V_0,C))}{V(Sp)}$ goes to zero when $C>0$.

Since Remark II is a negative news for large size data. In this
case, we can cut the total vector into shorter ones, with the
compensation on more bits to denote which bases have been used. The
next conjecture try to support MUB globally, where $m_i$,$u_i$ has
the same meaning.
\newtheorem{conjecture}{Conjecture}
\begin{conjecture}
When $X$ is a normalized uniform random vector, then:
\begin{equation}
E(max(m_1,m_2,...m_k))\geq E(max(u_1,u_2,...u_k))
\end{equation}
\end{conjecture}
Numerical analysis by the author strongly support the conjecture.

When the autocorrelation matrixes $Rx$ of $X$ is known, Principal
Component Analysis(PCA) \cite{PCA,PCA2} really works well. However,
it's hard to change the PCA base when $Rx$ is changed. It's
interesting to consider MUB when Rx is known, and choose the
unbiased bases following the information of the complete spectra of
$Rx$. As the discussion above, we could treat $X$ a bunch of
independent random vectors from different domains. So engineers only
need to choose the bases which have nice spectra to get an average
optimization. Theorem $5$ implies that some inaccuracy about the
autocorrelation matrixes won't effect much. But theorem $3$ says
that there must be some MUB spectra of $Rx$ looks bad.


\begin{center}
\vskip 1em
 VIII. CONCLUSIONS
 \vskip 1em
\end{center}
\vskip 1em In this paper, we studied the subject about how to
analyze, process, and utilize the information in the second order
moments between random variables. We presented a number of
applications of this subject.
 However, many problems remain open, and we list some important ones here:

(i) What about the information in moments of order higher than $2$?

(ii)How do we find MUB when $d$ is not power of prime?
In particular, for prime dimension $d$, there are simple formulas to
compute MUB and has fast algorithm to do transformation, what can we say about
the case when $d$ is not prime?

(iii) What about the physical meaning of the nonfourier bases?

(iv)What kind of second moments filter(see subsection $C$ of Section
$III$) are physical realizable?

We should noticed that Symmetric Informationally Complete Sets
(SICs)\cite{SICFir,SICEXT} can do a similar job. But we don't know
whether SICs exists for dimension larger than $45$ complex linear
space. It should be interesting to ask which one (SICs or MUB) is
more fundamental to express discrete statistic signals.
\begin{center}
\vskip 1em
 IX. ACKNOWLEDGEMENTS
 \vskip 1em
\end{center}
\vskip 1em
 Hongyi Yao thanks Dr. Xiaoming Sun for introducing the
notation of MUB. Hongyi Yao thanks Dr. Xiaoming Sun,  Feng Han from
WIST Lab, Dr. Hao Zhang, Dr. Hui Zhou, and Dr. Xiao Zhang for
discussion. Hongyi Yao thanks Pro.Ker-I Ko for suggestions in
English writing. In particular, Hongyi Yao thanks Jue Wang for
generous support.
\bibliography{YaoG}

\bibliographystyle{plain}
\end{document}